\documentclass[10pt,twoside,english]{article}
\usepackage[T1]{fontenc}
\usepackage[latin9]{inputenc}
\usepackage[letterpaper]{geometry}
\usepackage{xcolor}
\usepackage{units}
\usepackage{amsthm}
\usepackage{amsmath}
\usepackage{graphicx}
\usepackage{setspace}
\usepackage[all]{xy}
\PassOptionsToPackage{normalem}{ulem}
\usepackage{ulem}

\makeatletter

\providecolor{lyxadded}{rgb}{0,0,1}
\providecolor{lyxdeleted}{rgb}{1,0,0}

\newcommand{\lyxaddress}[1]{
\par {\raggedright #1
\vspace{1.4em}
\noindent\par}
}
\theoremstyle{plain}
\newtheorem{thm}{\protect\theoremname}
  \theoremstyle{plain}
  \newtheorem{lem}[thm]{\protect\lemmaname}

\usepackage[all]{xy}

\usepackage{multirow}

\def\frontmatter@abstractheading{}

\newcommand{\Force}{\mathsf{Force}}
\date{}

\makeatother

\usepackage{babel}
  \providecommand{\lemmaname}{Lemma}
\providecommand{\theoremname}{Theorem}

\begin{document}

\title{No-Forcing and No-Matching Theorems for Classical Probability Applied
to Quantum Mechanics}

\author{Ehtibar N. Dzhafarov\textsuperscript{1} and Janne V. Kujala\textsuperscript{2}}

\maketitle

\lyxaddress{\begin{center}
\textsuperscript{1}Purdue University, ehtibar@purdue.edu \\\textsuperscript{2}University
of Jyv\"askyl\"a, Department of Mathematical Information Technology,
jvk@iki.fi
\par\end{center}}
\begin{abstract}
\mbox{}

Correlations of spins in a system of entangled particles are inconsistent
with Kolmogorov's probability theory (KPT), provided the system is
assumed to be non-contextual. In the Alice-Bob EPR paradigm, non-contextuality
means that the identity of Alice's spin (i.e., the probability space
on which it is defined as a random variable) is determined only by
the axis $\alpha_{i}$ chosen by Alice, irrespective of Bob's axis
$\beta_{j}$ (and vice versa). Here, we study contextual KPT models,
with two properties: (1) Alice's and Bob's spins are identified as
$A_{ij}$ and $B_{ij}$, even though their distributions are determined
by, respectively, $\alpha_{i}$ alone and $\beta_{j}$ alone, in accordance
with the no-signaling requirement; and (2) the joint distributions
of the spins $A_{ij},B_{ij}$ across all values of $\alpha_{i},\beta_{j}$
are constrained by fixing distributions of some subsets thereof. Of
special interest among these subsets is the set of probabilistic connections,
defined as the pairs $\left(A_{ij},A_{ij'}\right)$ and $\left(B_{ij},B_{i'j}\right)$
with $\alpha_{i}\not=\alpha_{i'}$ and $\beta_{j}\not=\beta_{j'}$
(the non-contextuality assumption is obtained as a special case of
connections, with zero probabilities of $A_{ij}\not=A_{ij'}$ and
$B_{ij}\not=B_{i'j}$). Thus, one can achieve a complete KPT characterization
of the Bell-type inequalities, or Tsirelson's inequalities, by specifying
the distributions of probabilistic connections compatible with those
and only those spin pairs $\left(A_{ij},B_{ij}\right)$ that are subject
to these inequalities. We show, however, that quantum-mechanical (QM)
constraints are special. No-forcing theorem says that if a set of
probabilistic connections is not compatible with correlations violating
QM, then it is compatible only with the classical-mechanical correlations.
No-matching theorem says that there are no subsets of the spin variables
$A_{ij},B_{ij}$ whose distributions can be fixed to be compatible
with and only with QM-compliant correlations. 

\mbox{}$ $

\textsc{Keywords:} CHSH inequalities; contextuality; EPR/Bohm paradigm;
Fine's theorem; joint distribution; probabilistic couplings; probability
spaces; random variables; Tsirelson inequalities.

\markboth{Dzhafarov and Kujala}{No-Forcing and No-Matching Theorems}
\end{abstract}

\section{Introduction}

\begin{figure}
\begin{centering}
\includegraphics[scale=0.55]{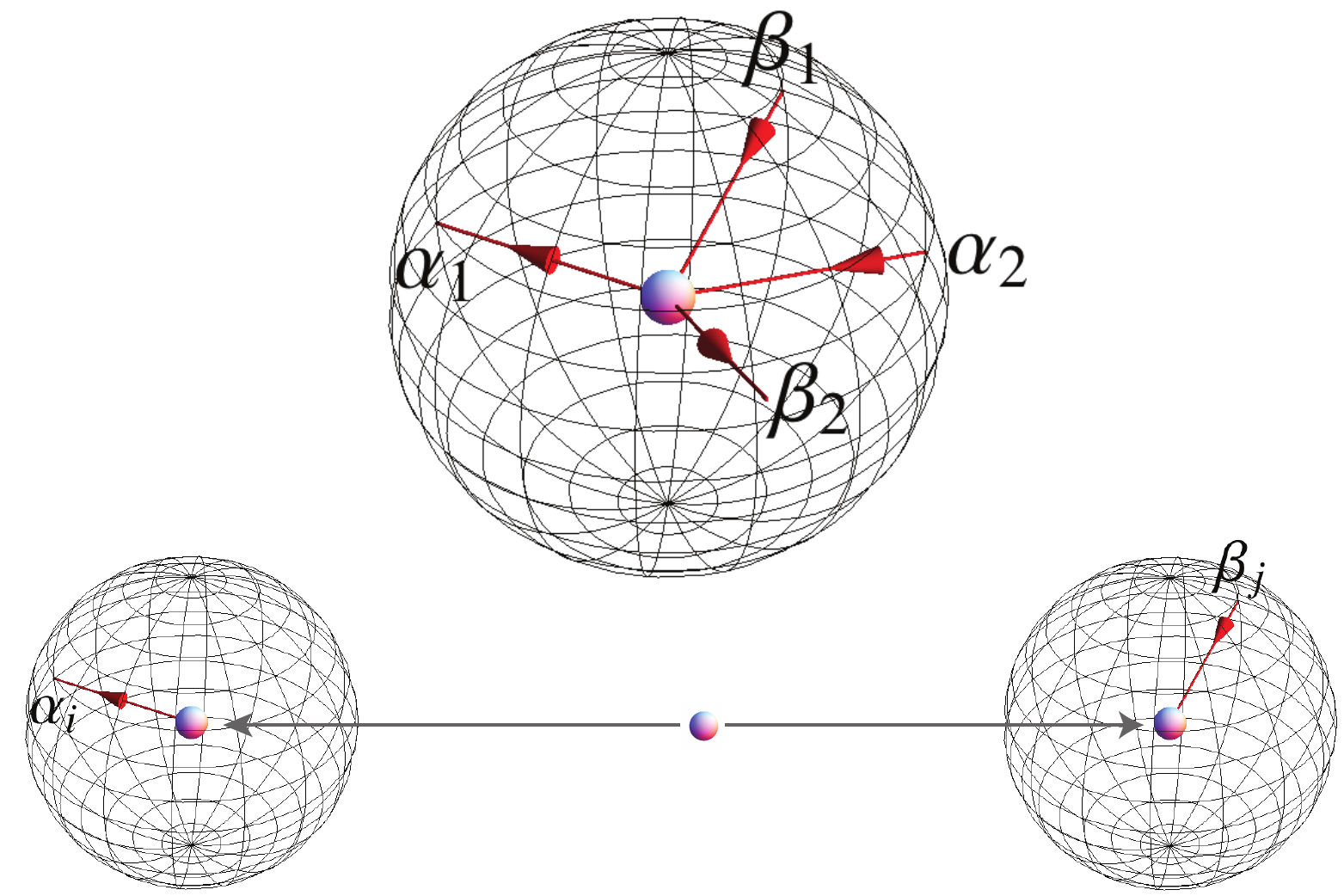}
\par\end{centering}

\caption[.]{The experimental paradigm considered in this paper. Two spin-half
particles created in a singlet state are running away from each other.
Each particle has its spin measured along one of two axes (measurement
settings): $\alpha_{1}$ or $\alpha_{2}$ for ``Alice's'' particle
(left), and $\beta_{1}$ or $\beta_{2}$ for ``Bob's'' particle
(right). Each measurement results in a random variable attaining one
of two values, $+1$ (spin-up, shown by outward-pointing cones) or
$-1$ (spin-down, inward-pointing cones). We confine the consideration
to the case when these values are equiprobable (so the no-signaling
requirement is satisfied trivially).\label{fig:Schematic-representation-of}}
\end{figure}

\textcolor{black}{Half a century }ago John Bell {[}\emph{\ref{enu:Bell,-J.-(1964).}}{]}
posed and answered in the negative the question of whether probability
distributions of spins in entangled particles could be accounted for
by a model written in the language of classical probability. These
distributions being among the most basic predictions of QM, Bell's
theorem and its subsequent elaborations {[}\emph{\ref{enu:J.F.-Clauser,-M.A.}-\ref{enu:Fine,-A.-(1982b).}}{]}
seem to mathematically isolate quantum determinism from the probabilistic
forms of classical determinism, and establish the necessity for a
quantum probability theory that is not reducible to the classical
one. However, Bell-type theorems do not engage the full potential
of the classical probability theory, if the latter is understood as
the theory adhering to Kolmogorov's conceptual framework {[}\emph{\ref{enu:A.-Kolmogorov,-Foundations}}{]}.
The use of probability theory in Bell-type theorems is constrained
by the following assumption:\\

\hangindent=15pt\hangafter=1\noindent(\emph{Non-Contextuality}, NC)
A spin of a given particle is a random variable whose \emph{identity}
does not depend on measurement settings (axes) chosen for other particles.\\ 

\noindent \textcolor{black}{This meaning of NC differs from the descriptions
and definitions found in the literature {[}}\textcolor{black}{\emph{\ref{enu:S.-Kochen,-F.}-\ref{enu:A.-Cabello:-Simple}}}\textcolor{black}{{]},
but all of them agree that Bell-type theorems are predicated on NC.
To understand our definition and why it leads to the Bell-type theorems,
we need to recapitulate basic facts about random variables in KPT.
Although the discussion can be conducted on a very high level of generality
{[}}\textcolor{black}{\emph{\ref{enu:E.N.-Dzhafarov,-J.V.PLOS}-\ref{Dzhafarov,-E.N.,-&Advances}}}\textcolor{black}{{]},
in this paper we confine it to the simplest Bohmian version of the
EPR paradigm {[}}\textcolor{black}{\emph{\ref{enu:D.-Bohm,-&}}}\textcolor{black}{{]},
depicted in Figure$\:$\ref{fig:Schematic-representation-of}. }For
each of the four combined settings $\left(\alpha_{i},\beta_{j}\right)$,
the recorded spins form a random pair $\left(A,B\right)$. \emph{No-signaling
requirement} (forced by special relativity if the two particles are
separated by a space-like interval, but usually assumed to hold even
if they are not) means that the distribution of $A$ does not depend
on $\beta_{j}$, nor the distribution of $B$ on $\alpha_{i}$: then,
by observing successive realizations of spin $A$ for a given value
of $\alpha_{i}$, Alice should never be able to guess that Bob exists. 

The no-signaling requirement is trivially satisfied if both $A$ and
$B$ are binary ($+1/-1$) random variables with 
\begin{equation}
\Pr\left[A=1\right]=\Pr\left[B=1\right]=\frac{1}{2},\label{eq:homogeneous}
\end{equation}
for all settings $\left(\alpha_{i},\beta_{j}\right)$.\textcolor{black}{{}
This is the case we confine our analysis to. }The joint distribution
of $A$ and $B$ for a given $\left(\alpha_{i},\beta_{j}\right)$
is then uniquely determined by the joint probability
\begin{equation}
p_{ij}=\Pr\left[A=1,B=1\right].\label{eq:joint}
\end{equation}

The distribution of a random variable, however, does not determine
its identity. In KPT, a binary ($+1/-1$) random variable $X$ is
identified with a mapping $f:S\rightarrow\left\{ -1,1\right\} $,
measurable with respect to some probability space $\left(S,\Sigma,\mu\right)$,
where $S$ is a set, $\Sigma$ is a set of events included in $S$,
and $\mu$ is a probability measure. The measurability of $f$ means
that $f^{-1}\left(1\right)\in\Sigma$, and the distribution of $X$
is defined by $\Pr\left[X=1\right]=\mu\left(f^{-1}\left(1\right)\right)$.
We say in this case that this random variable is \emph{defined on}
$\left(S,\Sigma,\mu\right)$. If two +1/-1 random variables, $X$
and $Y$, are \emph{jointly distributed}, then they are identified
with mappings $f:S\rightarrow\left\{ -1,1\right\} $ and $g:S\rightarrow\left\{ -1,1\right\} $,
measurable with respect to one and the same probability space $\left(S,\Sigma,\mu\right)$.
The joint distribution of these random variables is then determined
by
\begin{equation}
\Pr\left[X=\pm1,Y=\pm1\right]=\mu\left(f^{-1}\left(\pm1\right)\cap g^{-1}\left(\pm1\right)\right).
\end{equation}

The NC assumption says that the identity of $A$ at a given $\alpha_{i}$
does not depend on $\beta_{j}$, nor does the identity of $B$ at
a given $\beta_{j}$ depend on $\alpha_{i}$. If so, for any given
$\left(\alpha_{i},\beta_{j}\right)$, the output pair should be indexed
$\left(A_{i},B_{j}\right)$. Being jointly distributed, $A_{1}$ and
$B_{1}$, are defined on the same probability space, and so are $A_{1},B_{2}$,
and $A_{2},B_{1}$, and $A_{2},B_{2}$. It follows that all four random
variables $A_{1},A_{2},B_{1},B_{2}$ are defined on one and the same
$\left(S,\Sigma,\mu\right)$. That is, they are jointly distributed,
even though the joint distribution of $A_{1},A_{2}$ is not observable
in the sense of physical co-occurrence of their values, and analogously
for $B_{1},B_{2}$. The complete version of the Bell theorem {[}\emph{\ref{enu:Fine,-A.-(1982b).}}{]}
says that $A_{1},A_{2},B_{1},B_{2}$ with known distributions of $\left(A_{i},B_{j}\right)$
for $i,j\in\left\{ 1,2\right\} $ can be imposed a joint distribution
upon (i.e., $A_{1},A_{2},B_{1},B_{2}$ can be presented as measurable
functions on the same probability space) if and only if these marginal
distributions satisfy the CHSH inequalities,

\begin{equation}
0\leq p_{11}+p_{12}+p_{21}+p_{22}-2p_{ij}\leq1,\textnormal{ for all }i,j\in\left\{ 1,2\right\} ,\label{eq:Bell/CH}
\end{equation}
with $p_{ij}$ defined in (\ref{eq:joint}). In the general case we
have to replace the CHSH inequalities in the previous statement with
the conjunction of these inequalities and the no-signaling condition,
but we deal here with the special case (\ref{eq:homogeneous}), where
no-signaling is satisfied ``automatically.''

To relate this to the original language of the EPR discussions, the
statement that $A_{1},A_{2},B_{1},B_{2}$ are jointly distributed
(which is equivalent to NC) is equivalent to the ``hidden-variable''
assumption. Indeed, $A_{1},A_{2},B_{1},B_{2}$ are respectively representable
by functions $f_{1},f_{2},g_{1},g_{2}$ defined on the same $\left(S,\Sigma,\mu\right)$
if and only if 
\begin{equation}
A_{i}=f_{i}\left(R\right),\; B_{j}=g_{j}\left(R\right),\textnormal{ for all }i,j,\in\left\{ 1,2\right\} ,\label{eq:f-g}
\end{equation}
where $R$ is the random variable represented by the identity mapping
$\iota:S\rightarrow S$, $\iota\left(x\right)=x$, obviously measurable
with respect $\left(S,\Sigma,\mu\right)$. Moreover {[}\emph{\ref{enu:E.N.-Dzhafarov,-J.V.}}{]},
insofar as only binary $A_{1},A_{2},B_{1},B_{2}$ are concerned, the
``hidden variable'' $R$ can always be chosen to be a random variable
whose values are 16 possible combinations $\left(a_{1},a_{2},b_{1},b_{2}\right)=\left(\pm1,\pm1,\pm1,\pm1\right)$,
such that 
\begin{equation}
\Pr\left[R=\left(a_{1},a_{2},b_{1},b_{2}\right)\right]=\Pr\left[A_{1}=a_{1},A_{2}=a_{2},B_{1}=b_{1},B_{2}=b_{2}\right].\label{eq:f-g1}
\end{equation}
The functions $f_{i},g_{j}$ in (\ref{eq:f-g}) are then simply coordinate-wise
projections of the vectors $\left(a_{1},a_{2},b_{1},b_{2}\right)$:
\begin{equation}
f_{i}\left(a_{1},a_{2},b_{1},b_{2}\right)=a_{i},\; g_{j}\left(a_{1},a_{2},b_{1},b_{2}\right)=b_{j},\quad i,j\in\left\{ 1,2\right\} .\label{eq:f-g2}
\end{equation}

Summarizing, we have the equivalences
\begin{equation}
\begin{array}{c}
\textnormal{NC}\\
\Updownarrow\\
A_{1},A_{2},B_{1},B_{2}\textnormal{ are single-indexed }\\
\Updownarrow\\
\textnormal{ \ensuremath{A_{1},A_{2},B_{1},B_{2}\textnormal{ are jointly distributed }}}\\
\Updownarrow\\
A_{1},A_{2},B_{1},B_{2}\textnormal{ are functions of a \textquotedblleft hidden\textquotedblright\ }R\\
\Updownarrow\\
\textnormal{CHSH inequalities (\ref{eq:Bell/CH})},
\end{array}\label{eq:equivalences}
\end{equation}
of which the equivalence of the first four statements holds in KPT
essentially by definition. We know that the QM prediction for $p_{ij}$
in (\ref{eq:joint}) is, for $i,j\in\left\{ 1,2\right\} $,
\begin{equation}
p_{ij}=\nicefrac{1}{4}-\nicefrac{1}{4}\langle\alpha_{i}|\beta_{j}\rangle,\label{eq:cosine law}
\end{equation}
where $\langle\alpha_{i}|\beta_{j}\rangle$ is the cosine of the angle
between axes $\alpha_{i}$ and $\beta_{j}$, and we know that for
some choices of the axes these values of $p_{ij}$ violate (\ref{eq:Bell/CH}).
There are only two ways of dealing with this situation: to reject
KPT (replace it with a different, QM probability theory) or, if one
wishes to remain within the confines of KPT, to reject NC. The standard
QM theory has chosen the first way, we in this paper (following others,
e.g., {[}\emph{\ref{enu:A.-Yu.-Khrennikov,},\ref{enu:A.-Yu.-Khrennikov,-1}}{]})
explore the limits of the second one.

It is clear from (\ref{eq:equivalences}) that to reject NC means
to double-index $A$ or $B$ (by symmetry, $A$ and $B$). In the
Alice-Bob-paradigm considered this yields eight random variables $A_{ij}$
and $B_{ij}$, $i,j\in\left\{ 1,2\right\} $, with known joint distributions
for four pairs $\left(A_{ij},B_{ij}\right)$, 
\begin{equation}
\begin{array}{|c|c|c|}
\hline \left(\alpha_{i},\beta_{j}\right) & B_{ij}=+1 & B_{ij}=-1\\
\hline A_{ij}=+1 & p_{ij} & \nicefrac{1}{2}-p_{ij}\\
\hline A_{ij}=-1 & \nicefrac{1}{2}-p_{ij} & p_{ij}
\\\hline \end{array}\label{eq:observable probs}
\end{equation}
Since $A_{i1}$ and $A_{i2}$ are different (even if identically distributed)
random variables, and so are $B_{1j}$ and $B_{2j}$, we are no longer
forced to assume, as we were under NC, that all the random variables
in play are defined on one and the same probability space. In fact,
if we interpret a joint distribution as implying that values of jointly
distributed random variables can be observed ``together'' (e.g.,
simultaneously, in some inertial frame of reference), then $A_{i1}$
and $A_{i2}$ are not jointly distributed (and neither are $B_{1j}$
and $B_{2j}$). Our denial of NC, therefore, does not amount to admission
of a ``spooky action at a distance.'' Rather it is based on our
acknowledging, as a general principle,\\

\hangindent=\parindent\hangafter=1\noindent(\emph{Contextuality-by-Default},
CbD) No two spins recorded under different, mutually exclusive measurement
settings (across all particles involved) ever \emph{co-occur}, because
of which they are \emph{stochastically unrelated} (i.e., defined on
different probability spaces, possess no joint distribution) {[}\emph{\ref{enu:E.N.-Dzhafarov,-J.V.PLOS}-\ref{Dzhafarov,-E.N.,-&Advances}}{]}.\\

\noindent \textcolor{black}{The concept of four stochastically unrelated
to each other random pairs $\left(A_{ij},B_{ij}\right)$ is well within
the framework of KPT. Put differently, $A_{ij}$ and $B_{ij}$, being
jointly distributed for any given $\left(\alpha_{i},\beta_{j}\right)$,
can be considered functions of one and the same ``hidden'' random
variable $R_{ij}$, but KPT does not compel the four $R_{ij}$'s to
be viewed as jointly distributed. This seems to be the essence of
Gudder's analysis {[}}\textcolor{black}{\emph{\ref{enu:S.-Gudder:-On}}}\textcolor{black}{{]}
of the hidden-variable theories.} The existence of a single probability
space for all random variables imaginable is often mistakenly taken
as one of the tenets of Kolmogorov's theory. The untenability of this
view is apparent by cardinality considerations alone. Also, it is
easy to see that for any class of random variables of a given type
(e.g., binary +1/-1 ones) one can construct a variable of the same
type stochastically independent of each of them. The idea that all
such variables can be defined on a single space therefore leads to
a contradiction {[}\emph{\ref{enu:E.N.-Dzhafarov,-J.V.Qualified}}{]}. 

CbD does not mean, of course, that stochastically unrelated (i.e.,
possessing no joint distribution) spins cannot be mathematically \emph{imposed
}a joint distribution on (in the same way as dealing with spatially
disparate points on a sheet of paper does not prevent one from variously
grouping them in one's mind). In fact this can generally be done in
a variety of ways, referred to as different \emph{probabilistic couplings}
{[}\emph{\ref{enu:H.-Thorisson,-Coupling,}}{]} (p-couplings): \\

\hangindent=\parindent\hangafter=1\noindent(\emph{All-Possible-p-Couplings},
APpC) The stochastically unrelated spins recorded under different,
mutually exclusive measurement settings can be p-coupled arbitrarily,
insofar as the joint distribution imposed on them is consistent with
the observable joint distributions of the spins recorded under the
same measurement settings {[}\emph{\ref{enu:E.N.-Dzhafarov,-J.V.PLOS}-\ref{Dzhafarov,-E.N.,-&Advances}}{]}.\\

\noindent To p-couple the eight random variables $A_{ij},B_{ij}$,
$i,j,\in\left\{ 1,2\right\} $, means to specify $2{}^{8}$ probabilities
(nonnegative and summing to 1) 
\begin{equation}
\Pr\left[A_{11}=\pm1,\ldots,A_{22}=\pm1,B_{11}=\pm1,\ldots,B_{22}=\pm1\right].\label{eq:64}
\end{equation}
Each p-coupling is a way of designing a scheme of grouping realizations
of the eight random variables, \emph{as if} they co-occurred in an
imaginary experiment. Equivalently, to p-couple all $A_{ij},B_{ij}$
means to find a ``hidden'' random variable $R^{*}$ which all $A_{ij},B_{ij}$
can be presented as functions of,
\begin{equation}
A_{ij}=f_{ij}\left(R^{*}\right),\; B_{ij}=g_{ij}\left(R^{*}\right),\; i,j\in\left\{ 1,2\right\} .
\end{equation}
In particular, once the $2{}^{8}$ probabilities in (\ref{eq:64})
have been assigned to all vectors
\begin{equation}
\left(a_{11},\ldots,a_{22},b_{11},\ldots,b_{22}\right)=\left(\pm1,\ldots,\pm1,\pm1,\ldots,\pm1\right),
\end{equation}
one can choose $R^{*}$ to be the random variable with these values
and these probabilities {[}\emph{\ref{enu:E.N.-Dzhafarov,-J.V.}}{]},
the functions $f_{ij},g_{ij}$ being the coordinate-wise projections,
\begin{equation}
f_{ij}\left(a_{11},\ldots,a_{22},b_{11},\ldots,b_{22}\right)=a_{ij},\; g_{ij}\left(a_{11},\ldots,a_{22},b_{11},\ldots,b_{22}\right)=b_{ij},\; i,j\in\left\{ 1,2\right\} .
\end{equation}
Therefore, the ``hidden-variable'' meaning for a p-coupling of the
double-indexed $A,B$ is precisely the same as in the case of the
single-indexed ones, cf. (\ref{eq:f-g})-(\ref{eq:f-g2}). 

In fact, a p-coupling of the single-indexed $A_{1},A_{2},B_{1},B_{2}$
is merely a special case of p-couplings for the double-indexed $A_{11},\ldots,A_{22},B_{11},\ldots,B_{22}$.
One obtains this special case by constraining $H$ in (\ref{eq:coupling vector})
not only by the four empirically observable marginal distributions
of $\left(A_{ij},B_{ij}\right)$, but also by the additional assumption
\begin{equation}
\Pr\left[A_{i1}\not=A_{i2}\right]=0,\;\Pr\left[B_{1j}\not=B_{2j}\right]=0,\; i,j\in\left\{ 1,2\right\} .\label{eq:0-connections}
\end{equation}
Without this additional assumption a p-coupling for $A_{11},\ldots,A_{22},B_{11},\ldots,B_{22}$
always exists, whatever the observed distributions of $\left(A_{ij},B_{ij}\right)$.
For instance, one can always construct a p-coupling in which all stochastically
unrelated random pairs $\left(A_{ij},B_{ij}\right)$ are considered
stochastically independent. This is not particularly interesting,
precisely because this scheme of p-coupling is compatible with any
distributions of $\left(A_{ij},B_{ij}\right)$, QM-compliant and QM-contravening
alike.

An interesting question is whether there is a scheme by which the
stochastically unrelated random pairs $\left(A_{ij},B_{ij}\right)$
can be p-coupled so as to ``match'' QM precisely, in the sense of
allowing for all QM-compliant correlations and no other. At the end
of this paper we answer this question in the negative for all Kolmogorovian
models in which the p-couplings mentioned in APpC are constrained
by fixing distributions of some subsets of all spins involved (no-matching
theorem). 

Prior to that, however, we consider a case when these constraining
distributions are those of certain spin pairs, called \emph{connections}
{[}\emph{\ref{enu:E.N.-Dzhafarov,-J.V.PLOS}}{]}. This special class
of models is the most straightforward generalization of the models
compatible with NC. For the models with connections we prove a stronger
result (no-forcing theorem): such a model either allows for correlations
forbidden by QM or it only allows for the correlations of classical
mechanics, those satisfying the CHSH inequalities (\ref{eq:Bell/CH}).

\section{No-Forcing and No-Matching Theorems for QM}

The results of the Alice-Bob experiment depicted in Figure \ref{fig:Schematic-representation-of}
are uniquely described by the \emph{outcome vector} 
\begin{equation}
p=\left(p_{11},p_{12},p_{21},p_{22}\right).
\end{equation}
We say that $p$ is \emph{QM-compliant} if there exists some choice
of the settings $\alpha_{1},\alpha_{2},\beta_{1},\beta_{2}$ under
which $p$ satisfies (\ref{eq:cosine law}). The following inequality
is known to be a necessary and sufficient condition for $p$ being
QM-compliant {[}\emph{\ref{enu:Landau-LJ-(1988)}-\ref{enu:Kujala,-J.-V.,2008}}{]}:
\begin{equation}
\left|r_{11}r_{12}-r_{21}r_{22}\right|\leq\sqrt{1-r_{11}^{2}}\sqrt{1-r_{12}^{2}}+\sqrt{1-r_{21}^{2}}\sqrt{1-r_{22}^{2}},\label{eq:cosphericity}
\end{equation}
where 
\begin{equation}
r_{ij}=4p_{ij}-1\label{eq:correlation}
\end{equation}
is correlation between $A_{ij}$ and $B_{ij}$ ($i,j\in\left\{ 1,2\right\} $).
For geometric reasons obvious from Figure$\:$\ref{fig:Schematic-representation-of},
it is referred to as the \emph{cosphericity inequality} {[}\emph{\ref{enu:Kujala,-J.-V.,2008}}{]}.

In accordance with CbD, the spin pairs recorded under mutually exclusive
settings, say $\left(A_{11},B_{11}\right)$ and $\left(A_{12},B_{12}\right)$,
should generally be treated as stochastically unrelated\emph{,} possessing
no joint distribution. In accordance with APpC, one can consider all
possible eight-component random vectors
\begin{equation}
H=\left(A_{11},B_{11},A_{12},B_{12},A_{21},B_{21},A_{22},B_{22}\right)\label{eq:coupling vector}
\end{equation}
(see Figure$\:$\ref{fig:A-representation-of}) whose empirically
observable two-component parts $\left(A_{ij},B_{ij}\right)$ have
the distributions shown in the matrices (\ref{eq:observable probs}).%
\footnote{A rigorous formulation {[}\emph{\ref{enu:E.N.-Dzhafarov,-J.V.PLOS},\ref{enu:E.N.-Dzhafarov,-in press 2}-\ref{enu:E.N.-Dzhafarov,-&in press}}{]}
requires that $H$ be defined as $\left(A_{ij}',B_{ij}':i,j\in\left\{ 1,2\right\} \right)$
such that each pair $\left(A_{ij}',B_{ij}'\right)$ has the same distribution
as (rather than is identical to) $\left(A_{ij},B_{ij}\right)$ for
$i,j\in\left\{ 1,2\right\} $. Our lax notation is unlikely to cause
confusion in the present paper.\label{fn:A-rigorous-formulation}%
} A subset of any $k\leq8$ components of $H$ is referred to as its
$k$\emph{-marginal}, and its distribution is referred to as a $k$\emph{-marginal}
\emph{distribution} (the number $k$ being omitted if clear from the
context). So far we considered 1-marginals, that we posited to have
equiprobable $+1/-1$ values, and certain 2-marginals. Of the latter,
the marginal distributions of $\left(A_{ij},B_{ij}\right)$ are the
mandatory constraints imposed on $H$, in fact underlying the definition
of $H$. 

We know that, in this conceptual framework, NC is equivalent to the
choice of the p-coupling scheme in which (\ref{eq:0-connections})
is satisfied. Under the assumption that all 1-marginal probabilities
are $\nicefrac{1}{2}$, this means that we have the following joint
probabilities for the four (empirically unobservable) 2-marginals
$\left(A_{i1},A_{i2}\right)$ and $\left(B_{1j},B_{2j}\right)$:

\begin{equation}
\begin{array}{c}
\begin{array}{|c|c|c|}
\hline  & A_{i2}=+1 & A_{i2}=-1\\
\hline A_{i1}=+1 & \nicefrac{1}{2} & 0\\
\hline A_{i1}=-1 & 0 & \nicefrac{1}{2}
\\\hline \end{array}\\
\\
\begin{array}{|c|c|c|}
\hline  & \textnormal{\ensuremath{B_{2j}}=+1} & B_{2j}=-1\\
\hline B_{1j}=+1 & \nicefrac{1}{2} & 0\\
\hline B_{1j}=-1 & 0 & \nicefrac{1}{2}
\\\hline \end{array}
\end{array}\label{eq:classical connection}
\end{equation}
We know that these joint probabilities have neither empirical nor
theoretical justification, because de facto $A_{i1}$ and $A_{i2}$
(or $B_{1j}$ and $B_{2j}$) are not jointly distributed. There is
therefore no prohibition against p-coupling them differently, so that
generally, 
\begin{equation}
\begin{array}{c}
\begin{array}{|c|c|c|}
\hline \mbox{} & A_{i2}=+1 & A_{i2}=-1\\
\hline A_{i1}=+1 & \nicefrac{1}{2}-\varepsilon_{i}^{1} & \varepsilon_{i}^{1}\\
\hline A_{i1}=-1 & \varepsilon_{i}^{1} & \nicefrac{1}{2}-\varepsilon_{i}^{1}
\\\hline \end{array}\\
\\
\begin{array}{|c|c|c|}
\hline  & B_{2j}=+1 & B_{2j}=-1\\
\hline B_{1j}=+1 & \nicefrac{1}{2}-\varepsilon{}_{j}^{2} & \varepsilon{}_{j}^{2}\\
\hline B_{1j}=-1 & \varepsilon{}_{j}^{2} & \nicefrac{1}{2}-\varepsilon{}_{j}^{2}
\\\hline \end{array}
\end{array}\label{eq:general connection}
\end{equation}
In other words, a general p-coupling $H$ allows us to replace (\ref{eq:0-connections})
with
\[
\Pr\left[A_{i1}\not=A_{i2}\right]=2\varepsilon_{i}^{1},\;\Pr\left[B_{1j}\not=B_{2j}\right]=2\varepsilon{}_{j}^{2},
\]
where $0\leq\varepsilon_{i}^{1},\varepsilon{}_{j}^{2}\leq\nicefrac{1}{2}$,
$i,j\in\left\{ 1,2\right\} $. 

We use the term \emph{connection} to refer to the 2-marginals $\left(A_{i1},A_{i2}\right)$
and $\left(B_{1j},B_{2j}\right)$\emph{. }The four connections are
uniquely characterized by the \emph{connection vector} 
\begin{equation}
\varepsilon=\left(\varepsilon_{1}^{1},\varepsilon_{2}^{1},\varepsilon_{1}^{2},\varepsilon_{2}^{2}\right),
\end{equation}
or the corresponding correlations 
\begin{equation}
r=\left(r_{1}^{1},r_{2}^{1},r_{1}^{2},r_{2}^{2}\right),
\end{equation}
where $r_{l}^{k}=1-4\varepsilon_{l}^{k}$, $k,l\in\left\{ 1,2\right\} $.

An outcome vector $p$ and a connection vector $\varepsilon$ are
\emph{mutually compatible} {[}\emph{\ref{enu:E.N.-Dzhafarov,-J.V.PLOS}}{]}
if they can be embedded in one and the same p-coupling $H$, as shown
in Figure$\:$\ref{fig:A-representation-of}. In other words, $\varepsilon$
and $p$ are mutually compatible if they can be computed as marginal
probabilities from the probabilities assigned to the $2^{8}$ values
of $H$. That is, $p_{ij}$ is the sum of the probabilities for all
values of $H$ with $A_{ij}=B_{ij}=1$, and $\varepsilon_{i}^{1},\varepsilon{}_{j}^{2}$
are the sums of the probabilities for all values of $H$ with, respectively,
$A_{i1}=-A_{i2}=1$ and $B_{1j}=-B_{2j}=1$. The set of all compatible
pairs $\left(p,\varepsilon\right)$ forms an 8-dimensional polytope
described by Lemma \ref{lem:p-and-}. But we need some notation first.

Given a connection vector $\varepsilon=\left(\varepsilon_{1}^{1},\varepsilon_{2}^{1},\varepsilon_{1}^{2},\varepsilon_{2}^{2}\right)$,
consider the sums 
\begin{equation}
\frac{1}{4}\left(\pm r_{1}^{1}\pm r_{2}^{1}\pm r_{1}^{2}\pm r_{2}^{2}\right)
\end{equation}
where each $\pm$ is replaced with either $+$ or $-$. Let $s_{0}\left(\varepsilon\right)$
denote the largest of the eight such sums with even numbers of plus
signs, 
\begin{equation}
s_{0}\left(\varepsilon\right)=\max\left\{ \frac{1}{4}\left(\pm r_{1}^{1}\pm r_{2}^{1}\pm r_{1}^{2}\pm r_{2}^{2}\right):\textnormal{ the number of }+\textnormal{'s is 0, 2, or 4}\right\} .\label{eq:s0}
\end{equation}
Let 
\begin{equation}
s_{1}\left(\varepsilon\right)=\max\left\{ \frac{1}{4}\left(\pm r_{1}^{1}\pm r_{2}^{1}\pm r_{1}^{2}\pm r_{2}^{2}\right):\textnormal{ the number of }+\textnormal{'s is 1 or 3}\right\} .\label{eq:s1}
\end{equation}
Since the components of $\varepsilon$ belong to $\left[0,\nicefrac{1}{2}\right]$,
the pairs $\left(s_{0}\left(\varepsilon\right),s_{1}\left(\varepsilon\right)\right)$
fill in the triangular area connecting $\left(0,0\right)$, $\left(\nicefrac{1}{2},1\right)$,
and $\left(1,\nicefrac{1}{2}\right)$. In particular,
\begin{equation}
\begin{array}{c}
s_{0}\left(\varepsilon\right)+s_{1}\left(\varepsilon\right)\leq\nicefrac{3}{2},\\
0\leq s_{0}\left(\varepsilon\right)\leq1,\\
0\leq s_{1}\left(\varepsilon\right)\leq1.
\end{array}\label{eq:properties}
\end{equation}

We define $s_{0}\left(p\right)$ and $s_{1}\left(p\right)$ for any
outcome vector $p=\left(p_{11},p_{12},p_{21},p_{22}\right)$ analogously
(using $r_{ij}=4p_{ij}-1$ in place of $r_{j}^{i}$, see Figure \ref{fig:A-representation-of}):
\begin{equation}
\begin{array}{c}
s_{0}\left(p\right)=\max\left\{ \frac{1}{4}\left(\pm r_{11}\pm r_{12}\pm r_{21}\pm r_{22}\right):\textnormal{ the number of }+\textnormal{'s is 0, 2, or 4}\right\} ,\\
\\
s_{1}\left(p\right)=\max\left\{ \frac{1}{4}\left(\pm r_{11}\pm r_{12}\pm r_{21}\pm r_{22}\right):\textnormal{ the number of }+\textnormal{'s is 1 or 3}\right\} .
\end{array}
\end{equation}
 Since the components of $p$ also belong to $\left[0,\nicefrac{1}{2}\right]$,
the pairs $s_{0}\left(p\right),$ $s_{1}\left(p\right)$ have precisely
the same properties as $s_{0}\left(\varepsilon\right),s_{1}\left(\varepsilon\right)$.

\begin{figure}
\[
\xymatrix{ & B_{12}\ar[rr]_{r_{2}^{2}=1-4\varepsilon_{2}^{2}}\ar[dl] &  & B_{22}\ar[ll]\ar[dr]\\
A_{12}\ar[d]\ar[ur]^{r_{12}=4p_{12}-1} &  &  &  & A_{22}\ar[d]_{r_{2}^{1}=1-4\varepsilon_{2}^{1}}\ar[ul]_{r_{22}=4p_{22}-1}\\
A_{11}\ar[u]_{r_{1}^{1}=1-4\varepsilon_{1}^{1}}\ar[dr]_{r_{11}=4p_{11}-1} &  &  &  & A_{21}\ar[u]\ar[dl]^{r_{21}=4p_{21}-1}\\
 & B_{11}\ar[rr]^{r_{1}^{1}=1-4\varepsilon_{1}^{2}}\ar[lu] &  & B_{21}\ar[ll]\ar[ur]
}
\]

\caption[.]{A p-coupling $H$ for pairs $\left(A_{ij},B_{ij}\right)$ (see footnote
\ref{fn:A-rigorous-formulation}). The number at a double-arrow connecting
two random variables is their correlation. Horizontal and vertical
arrows correspond to a connection vector $\varepsilon=\left(\varepsilon_{1}^{1},\varepsilon_{2}^{1},\varepsilon_{1}^{2},\varepsilon_{2}^{2}\right)$;
diagonal arrows correspond to an outcome vector $p=\left(p_{11},p_{12},p_{21},p_{22}\right)$.\label{fig:A-representation-of} }
\end{figure}
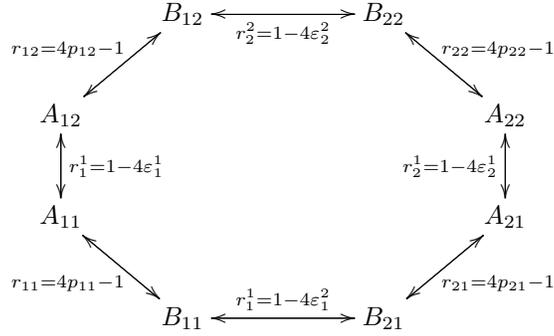

\begin{lem}
\label{lem:p-and-}p and $\varepsilon$ are mutually compatible if
and only if 
\begin{equation}
\begin{array}{c}
s_{0}\left(\varepsilon\right)+s_{1}\left(p\right)\le\nicefrac{3}{2},\\
s_{1}\left(\varepsilon\right)+s_{0}\left(p\right)\le\nicefrac{3}{2}.
\end{array}
\end{equation}

\end{lem}
\noindent For the proof of this lemma see {[}\emph{\ref{enu:E.N.-Dzhafarov,-J.V.PLOS}}{]}.

The following observation that we need later on is proved by simple
algebra.
\begin{lem}
\label{lem:E0}The set $E_{0}$ of connection vectors $\varepsilon$
with $s_{0}\left(\varepsilon\right)=1$ consists of the null vector
$\varepsilon_{0}=\left(0,0,0,0\right)$ and seven vectors obtained
by replacing any two of or all four zeros in $\varepsilon_{0}$ with
$\nicefrac{1}{2}$. For all $\varepsilon$ in $E_{0}$, $s_{1}\left(\varepsilon\right)=\nicefrac{1}{2}$.
\end{lem}
The null (or identity) connection vector $\varepsilon_{0}$ plays
a special role, as it corresponds to (\ref{eq:0-connections}) and
(\ref{eq:classical connection}): this is the choice implicitly made
in all Bell-type theorems. It also plays a central role in the no-forcing
theorem below. Note that according to Lemmas \ref{lem:p-and-} and
\ref{lem:E0}, an outcome vector $p$ is compatible with $\varepsilon_{0}$
if and only if it is compatible with all connection vectors in $E_{0}$. 

It is easy to see that a connection vector can be chosen so that it
is compatible with all QM-compliant outcome vectors. The simplest
example is the connection vector $\varepsilon_{ind}=\left(\nicefrac{1}{4},\nicefrac{1}{4},\nicefrac{1}{4},\nicefrac{1}{4}\right)$
corresponding to the p-couplings (\ref{eq:coupling vector}) with
all components pairwise independent, except, possibly, for pairs $\left(A_{ij},B_{ij}\right)$.
Since $s_{0}\left(\varepsilon_{ind}\right)=s_{1}\left(\varepsilon_{ind}\right)=0$,
it follows from Lemma \ref{lem:p-and-} that $\varepsilon_{ind}$
is compatible with any $p$, whether QM-compliant or not. 

What is less obvious is the answer to the question: what are all the
connection vectors that are compatible \emph{only} with QM-compliant
outcome vectors? In other words, we are interested in the set $\Force_{QM}$
of connection vectors, defined as follows:\\

\hangindent=\parindent\hangafter=1\noindent$\Force_{QM}$: all connection
vectors $\varepsilon$ such that if $p$ is compatible with $\varepsilon$,
then $p$ is QM-compliant, i.e., satisfies the cosphericity inequality
(\ref{eq:cosphericity}).\\

\noindent The name of the set is to indicate that $\varepsilon\in\Force_{QM}$
``forces'' every $p$ compatible with it to be QM-compliant. The
set is not empty, because, as the next lemma shows, it includes $E_{0}$
of Lemma \ref{lem:E0}.
\begin{lem}
\label{lem:Part of Force}$E_{0}\subset\Force_{QM}$.\end{lem}
\begin{proof}
By Fine's theorem {[}\emph{\ref{enu:Fine,-A.-(1982b).}}{]},%
\footnote{\textcolor{black}{The theorem states that the single-indexed $A_{1},A_{2},B_{1},B_{2}$
are jointly distributed if and only if $p$ satisfies the CHSH inequalities.
We use the fact that the single-indexation means that the connection
vector for the double-indexed $A$'s and $B$'s is $\varepsilon_{0}$,
and that the existence of the joint distribution of these $A$'s and
$B$'s means, by definition, that $\varepsilon_{0}$ and $p$ are
compatible.}%
} $p$ is compatible with $\varepsilon_{0}$ (hence, by Lemmas \ref{lem:p-and-}
and \ref{lem:E0}, also with other members of $E_{0}$) if and only
if it satisfies the CHSH inequalities (\ref{eq:Bell/CH}). Since the
cosphericity inequality is a necessary condition for the compatibility
of $p$ with $\varepsilon_{0}$ {[}\emph{\ref{enu:Kujala,-J.-V.,2008}}{]},
QM-compliance follows from the CHSH inequalities. 
\end{proof}
We are thus led to the following questions:\\ 

\hangindent=\parindent\hangafter=1\noindent(Q1) What is the entire
set $\Force_{QM}$ (what connection vectors it contains beside $E_{0}$)?\\ 

\hangindent=\parindent\hangafter=1\noindent(Q2) What is the set $P_{QM}$
of the outcome vectors $p$ each of which is compatible with at least
one of the connection vectors in $\Force_{QM}$?\\ 

\noindent The questions are significant for the following reason.
If $P_{QM}$ turned out to coincide with the set of all QM-compliant
$p$, we would have a hope of constructing a KPT model that would
match QM in the sense of allowing all those $p$ that are possible
in QM and forbidding all those $p$ that QM forbids. Quantum determinism
could then be ``explained'' by pointing out that the multiple probability
spaces corresponding to different measurement settings can be p-coupled
by using appropriately chosen connection vectors for different settings.
However, this hope should be abandoned, because $\Force_{QM}$ in
fact coincides with $E_{0}$, whence $P_{QM}$ includes only those
$p$ that satisfy the CHSH inequalities. It is well known that the
CHSH inequalities do not describe all QM-compliant vectors: e.g.,
they are violated if we use in (\ref{eq:cosine law}) coplanar vectors
at the angles $\alpha_{1}=0$, $\alpha_{2}=\nicefrac{\pi}{2}$, $\beta_{1}=\nicefrac{\pi}{4}$,
$\beta_{2}=-\nicefrac{\pi}{4}$.

The proof makes use of the following observation.
\begin{lem}
\label{lem:If-,-then}If $p$ belong to the set $P_{0}$ described
by
\[
\begin{array}{l}
s_{0}\left(p\right)+s_{1}\left(p\right)=\nicefrac{3}{2},\\
s_{0}\left(p\right)<1,
\end{array}
\]
then $p$ is not QM-compliant (violates the cosphericity inequality\textup{).}\end{lem}
\begin{proof}
One easily checks that
\[
\begin{array}{l}
s_{i}\left(p\right)=\nicefrac{1}{4}\left(\left|r_{11}\right|+\left|r_{12}\right|+\left|r_{21}\right|+\left|r_{22}\right|\right)\\
s_{1-i}\left(p\right)=s_{i}\left(p\right)-\nicefrac{1}{2}\min\left(\left|r_{11}\right|,\left|r_{12}\right|,\left|r_{21}\right|,\left|r_{22}\right|\right),
\end{array}
\]
where $i$ is 0 or 1 according as the number of positive correlations
$r_{ij}$ is even or odd. Without loss of generality, let the minimum
in the second expression equal $\left|r_{22}\right|$. Then 
\[
s_{0}\left(p\right)+s_{1}\left(p\right)=\nicefrac{1}{2}\left(\left|r_{11}\right|+\left|r_{12}\right|+\left|r_{21}\right|\right),
\]
and this can only equal $\nicefrac{3}{2}$ if each of the three correlations
equals $\pm1$. The cosphericity inequality (\ref{eq:cosphericity})
can only be satisfied then if 
\[
r_{22}=\pm1\textnormal{ and }r_{11}r_{12}=r_{21}r_{22}.
\]
It is easy to see that the latter is possible only if the number of
$+1$'s among the four $\pm1$ correlations is even. It follows that
$s_{0}\left(p\right)=1$. \end{proof}
\begin{thm}[no-forcing]
 The answer to Q1 is: $\Force_{QM}=E_{0}$ (whence the answer to
Q2 is: $P_{QM}$ is the set of all $p$ satisfying CHSH inequalities).\end{thm}
\begin{proof}
We know from (\ref{eq:properties}) that $s_{0}\left(\varepsilon\right)\leq1$,
and from Lemmas \ref{lem:E0} and \ref{lem:Part of Force} we know
that $s_{0}\left(\varepsilon\right)=1$ describes the set $E_{0}\subset\Force_{QM}$.
The theorem is proved by showing that $\Force_{QM}$ does not contain
any $\varepsilon$ with $s_{0}\left(\varepsilon\right)<1$. From the
definition of $\Force_{QM}$, if there is a $p$ with which $\varepsilon$
is compatible but which does not satisfy the cosphericity inequality
(\ref{eq:cosphericity}), then $\varepsilon\not\in\Force_{QM}$. By
Lemma \ref{lem:p-and-}, if for a given $\varepsilon$ one chooses
a $p$ such that
\[
\begin{array}{c}
s_{0}\left(\varepsilon\right)+s_{1}\left(p\right)\le s_{0}\left(p\right)+s_{1}\left(p\right)=\nicefrac{3}{2},\\
s_{1}\left(\varepsilon\right)+s_{0}\left(p\right)\le s_{0}\left(p\right)+s_{1}\left(p\right)=\nicefrac{3}{2},
\end{array}
\]
then $p$ and $\varepsilon$ are compatible. If $s_{0}\left(\varepsilon\right)\leq\nicefrac{1}{2}$,
then choose $p$ with $s_{1}\left(p\right)=1$ and $s_{0}\left(p\right)=\nicefrac{1}{2}$
to satisfy this system. If $\nicefrac{1}{2}<s_{0}\left(\varepsilon\right)<1$,
then choose $p$ with $s_{1}\left(p\right)=s_{1}\left(\varepsilon\right)$
and $s_{0}\left(p\right)=\nicefrac{3}{2}-s_{1}\left(\varepsilon\right)$
to satisfy this system. By Lemma \ref{lem:If-,-then}, all these choices
of $p$ belong to $P_{0}$ and therefore violate the cosphericity
inequality. It follows that all $\varepsilon$ with $s_{0}\left(\varepsilon\right)<1$
do not belong to $\Force_{QM}$. 
\end{proof}
One consequence of this theorem is that KPT in which p-couplings are
constrained by connections cannot \emph{match} QM precisely: if it
allows \emph{only} for QM-compliant $p$ (as does the choice of $\varepsilon=\varepsilon_{0}$),
then it only allows for a proper subset thereof; and if it allows
for \emph{all} QM-compliant $p$, then it also allows for some $p$
that are QM-contravening (as does the connection vector $\varepsilon_{ind}=\left(\nicefrac{1}{4},\nicefrac{1}{4},\nicefrac{1}{4},\nicefrac{1}{4}\right)$
that is compatible with all possible outcome vectors $p$). We now
generalize this no-matching statement for connection vectors to arbitrary
marginal distributions imposed on p-couplings $H$ in (\ref{eq:coupling vector}). 

\begin{figure*}
\begin{centering}
\includegraphics[scale=0.5]{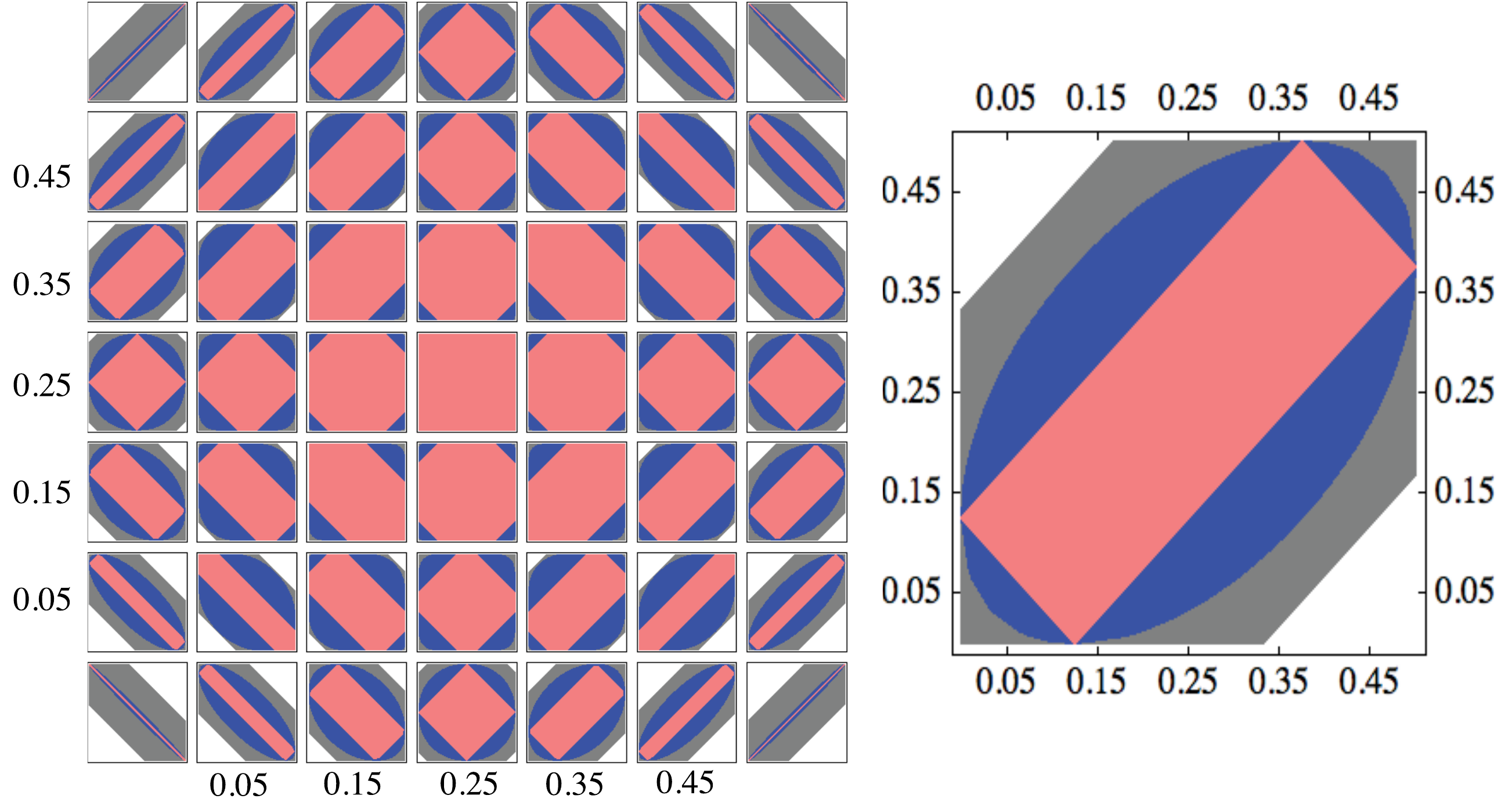}
\par\end{centering}

\caption[.]{Areas of outcome vectors $p=\left(p_{11},p_{12},p_{21},p_{22}\right)$,
two components of which (no matter which) define panels on the left,
and the remaining two form the axes of each panel, as shown on the
right. The pink area contains all $p$ satisfying the CHSH inequalities\emph{
}(\ref{eq:Bell/CH}). One can constrain the p-couplings $H$ by marginal
distributions (notably, by putting $\mbox{\ensuremath{\varepsilon}=\ensuremath{\left(0,0,0,0\right)}}$),
so that the set of all $p$ compatible with them coincides with the
pink area {[}\emph{\ref{enu:Bell,-J.-(1964).}-\ref{enu:Fine,-A.-(1982b).}}{]}.
The gray area contains all $p$ satisfying the Tsirelson inequalities
(\ref{eq:Tsirelson}). One can constrain the p-couplings $H$ by marginal
distributions (e.g., by putting $\varepsilon=\left(\nicefrac{\left(\sqrt{2}-1\right)}{8},\nicefrac{\left(\sqrt{2}-1\right)}{8},\nicefrac{\left(\sqrt{2}-1\right)}{8},\nicefrac{\left(\sqrt{2}-1\right)}{8}\right)$),
so that the set of all $p$ compatible with them coincides with the
gray area {[}\emph{\ref{enu:E.N.-Dzhafarov,-J.V.PLOS}}{]}. The blue
area (that includes the pink area and is included in the gray one)
contains all $p$ satisfying the cosphericity inequality (\ref{eq:cosphericity}),
i.e., all QM-compliant $p$. The no-matching theorem says that, for
any set of marginal distributions, the set of all $p$ compatible
with them never coincides with the blue area precisely.\label{fig:Areas-of-outcome}}
\end{figure*}

The empirically inaccessible 2-marginals $\left(A_{i1},A_{i2}\right)$
and $\left(B_{1j},B_{2j}\right)$ constrain $H$ by specifying a connection
vector $\varepsilon$. There are, however, other empirically inaccessible
marginals, such as $\left(A_{11},A_{22}\right)$, $\left(A_{11},A_{12},A_{21},A_{22}\right)$,
$\left(A_{11},B_{22},B_{12}\right)$, etc. Each of them, once its
distribution is specified, constrains the possible distributions of
$H$.
\begin{thm}[no-matching]
There is no set of marginal distributions imposed on $H$ such that
outcome vectors $p$ are compatible with this set if and only if they
are QM-compliant. \end{thm}
\begin{proof}
Observe first that a distribution of any $k$-marginal $\left(X_{1},\ldots,X_{k}\right)$
of $H$ can be presented by $2^{k}$ probabilities 
\begin{equation}
\Pr\left[X_{i_{1}}=1,\ldots,X_{i_{k'}}=1\right],\label{eq:all 1's}
\end{equation}
with all values equated to 1, for all $k'$-(sub)marginals of $X_{1},\ldots,X_{k}$
($0\leq k'\leq k$). This includes the empty subset, for which we
put $\Pr\left[\right]=1$. If we fix distributions of several marginals,
with the numbers of components $k_{1},\ldots,k_{m}$, then the total
number of different probabilities is $N<2^{k_{1}}+\ldots+2^{k_{m}}$.
This set of probabilities constrains the set of possible outcome vectors
$p$ to those for which one can find a $2^{8}$-component vector $Q$
with the following properties: 
\begin{equation}
\begin{array}{c}
M_{1}Q=p\\
\textnormal{subject to}\\
M_{2}Q=P,Q\geq0.
\end{array}\label{eq:linear programming}
\end{equation}
Here, $Q$ is the vector of probabilities assigned to all possible
values of $H$ (and $Q\geq0$ is understood componentwise), $M_{1},M_{2}$
are Boolean (0/1) matrices with dimensions $4\times2^{8}$ and $N\times2^{8}$,
respectively, and $P$ is the vector of all probabilities of the form
(\ref{eq:all 1's}) that define the distributions of the marginals
chosen. The entries of the matrices are defined by the following rule:
(1) choose the row of the matrix corresponding to the probability
$\Pr\left[X_{1}=1,\ldots,X_{k}=1\right]$; (2) choose the column of
$M$ corresponding to values
\[
\left(A{}_{ij}=a_{ij},B{}_{ij}=b_{ij}:i,j\in\left\{ 1,2\right\} \right)
\]
of $H$ ($a{}_{ij},b_{ij}\in\left\{ -1,1\right\} $); (3) put 1 in
the intersection of this row and this column if and only if $a_{ij}$
and $b_{ij}$ equal 1 for all $A_{ij}$ and $B{}_{ij}$ that belong
$\left(X_{1},\ldots,X_{k}\right)$; (4) for the 0-marginal (empty
set), all entries are 1. In matrix $M_{1}$ the marginals for its
four rows are $\left(A{}_{ij},B{}_{ij}\right)$, $i,j\in\left\{ 1,2\right\} $,
and the row corresponding to, e.g., $\left(A{}_{11},B{}_{11}\right)$
contains $1$ in each cell whose column corresponds to $H$-values
with 
\[
\left(\begin{array}{ccccc}
A_{11} & B{}_{11} & A{}_{12} & \ldots & B{}_{22}\\
1 & 1 & \textnormal{any} & \ldots & \textnormal{any}
\end{array}\right).
\]
To illustrate the structure of matrix $M_{2}$ and vector $P$, assume
that one of the marginals chosen is $\left(A_{11},A_{12},A_{21},A_{22}\right)$.
Then $P$ includes the 16 probabilities
\[
\begin{array}{c}
\Pr\left[A_{11}=1,A_{12}=1,A_{21}=1,A_{22}=1\right]\\
\Pr\left[A_{11}=1,A_{12}=1,A_{21}=1\right],\ldots,\Pr\left[A_{12}=1,A_{21}=1,A_{22}=1\right]\\
\Pr\left[A_{11}=1,A_{12}=1\right],\ldots,\Pr\left[A_{21}=1,A_{22}=1\right]\\
\Pr\left[A_{11}=1\right]=\nicefrac{1}{2},\ldots,\Pr\left[A_{22}=1\right]=\nicefrac{1}{2}\\
\Pr\left[\right]=1.
\end{array}
\]
The row of $M_{2}$ corresponding to, say, $\left(A_{11},A_{12},A_{21}\right)$,
contains 1 for all columns with $H$-values
\[
\left(\begin{array}{cccccccc}
A_{11} & B{}_{11} & A{}_{12} & B{}_{12} & A{}_{21} & B{}_{21} & A{}_{22} & B{}_{22}\\
1 & \textnormal{any} & 1 & \textnormal{any} & 1 & \textnormal{any} & \textnormal{any} & \textnormal{any}
\end{array}\right).
\]
Now, the set of all vectors $p$ for which a $Q$ exists satisfying
(\ref{eq:linear programming}) forms a polytope confined within a
$\left[0,\nicefrac{1}{2}\right]^{4}$ cube. This polytope can be empty
(if the distributions for the marginals chosen are not compatible),
consist of a single point (e.g., if the marginals chosen include $H$
itself), or have any dimensionality between 1 and 4. The statement
of the theorem follows from the fact that the set of QM-compliant
$p$, those satisfying (\ref{eq:cosphericity}), is not a polytope.
Figure \ref{fig:Areas-of-outcome} makes this fact obvious, by showing
the curvilinear shape of the two-dimensional cross-sections of the
set of QM-compliant $p$.
\end{proof}

\section{Conclusion}

We have seen that the Bell-type theorems do not allow one to gauge
the ability of KPT for dealing with QM-compliant spin distributions
in entangled particles. The Bell-type theorems are confined to the
NC assumption, and the latter is not an integral part of KPT. The
power of KPT is much greater if one uses its basic conceptual apparatus
to systematically distinguish the distribution of spins (which, due
to the no-signaling requirement, can never be affected by the measurement
settings chosen in distant particles) and the identity of spins as
random variables \textendash{}\textendash{}\textendash{} which, in
accordance with the CbD principle, may very well depend on the settings
chosen across all particles (without violating any known laws of physics).
Our results show, however, that QM preserves its special status even
at this, much greater level of generality. The contextual KPT models
(with marginal constraints) are not able to match QM predictions precisely. 

Figure \ref{fig:Areas-of-outcome} serves an additional purpose of
demonstrating that this failure of the contextual KPT models with
marginal constraints is not due to its general inability to match
theories outside the scope of classical mechanics. It is the nonlinearity
of the area of QM-compliant outcome vectors rather than their non-classicality
that is responsible for the no-matching theorem. Thus, consider the
\emph{Tsirelson inequalities} {[}\emph{\ref{enu:Cirel'son-BS-(1980)},\ref{enu:Landau-LJ-(1987)}}{]}
for spin-$\nicefrac{1}{2}$ particles (with equiprobable spin-up and
spin-down in all directions), 
\begin{equation}
\frac{1-\sqrt{2}}{2}\leq p_{11}+p_{12}+p_{21}+p_{22}-2p_{ij}\leq\frac{1+\sqrt{2}}{2}\quad\left(i,j\in\left\{ 1,2\right\} \right).\label{eq:Tsirelson}
\end{equation}
They are known to be satisfied by all QM-compliant outcome vectors
$p$, and they impose the lower and upper bound on the QM-permitted
violations of the CHSH inequalities: a linear combination $p_{11}+p_{12}+p_{21}+p_{22}-2p_{ij}$
can achieve the values $\frac{1-\sqrt{2}}{2}$ and $\frac{1+\sqrt{2}}{2}$
by appropriate choices of the directions $\alpha_{1},\alpha_{2},\beta_{1},\beta_{2}$
in (\ref{eq:cosine law}). 

Our approach allows one to offer a KPT account for the Tsirelson bounds
by postulating that, in the Alice-Bob system depicted in Figure \ref{fig:Schematic-representation-of},
the connection vectors $\varepsilon$ satisfy
\begin{equation}
s_{0}\left(\varepsilon\right)=\frac{3-\sqrt{2}}{2},\; s_{1}\left(\varepsilon\right)\leq\frac{1}{2},\label{eq:Tsirelson-compatible}
\end{equation}
where $s_{0}\left(\varepsilon\right)$ and $s_{1}\left(\varepsilon\right)$
are defined in (\ref{eq:s0})-(\ref{eq:s1}). It has been shown {[}\emph{\ref{enu:E.N.-Dzhafarov,-J.V.PLOS}}{]}
that such a connection vector is compatible with those and only those
outcome vectors $p$ that satisfy the Tsirelson inequalities (\ref{eq:Tsirelson}). 

Recall that a connection vector $\varepsilon$ is compatible with
those and only those $p$ that satisfy the CHSH inequalities (\ref{eq:Bell/CH})
if and only if $s_{0}\left(\varepsilon\right)=1.$ The latter means
that $A_{i1}$ and $A_{i2}$ in the classical-mechanical system are
either always equal or always opposite, and the same is true for $B_{1j}$
and $B_{2j}$. With $\varepsilon$ satisfying (\ref{eq:Tsirelson-compatible}),
$A_{i1}$ and $A_{i2}$, as well as $B_{1j}$ and $B_{2j}$, may be
unequal and non-opposite with some small probabilities, e.g., $\nicefrac{\left(\sqrt{2}-1\right)}{8}$,
if one assumes that these probabilities are the same for all four
connections $\left(A_{i1},A_{i2}\right)$, $\left(B_{1j},B_{2j}\right)$,
$i,j\in\left\{ 1,2\right\} $. If these probabilities were larger,
the connection vector would be compatible with outcome vectors exceeding
the Tsirelson bounds.

This is not, of course, a physical explanation, but a principled way
of embedding the Tsirelson bounds within the framework of KPT. The
connection probabilities themselves do not have an interpretation
within a physical theory. Recall, however, that connections are merely
2-marginals of a p-coupled eight-component vector $H$ in (\ref{eq:coupling vector}),
and that we use the connections to delineate a class of such p-couplings.
As we explained in Introduction, a p-coupling $H$ allows for the
same interpretation in terms of ``hidden variables'' as the one
traditionally used (whether or not one calls it ``physical'') in
the derivation and analysis of the Bell-type theorems. 

Recently, Cabello {[}\ref{enu:A.-Cabello:-Simple}{]} attempted to
find an account for the Tsirelson bounds using another principle.
In his analysis of the Alice-Bob paradigm, he considers sequences
of stochastically independent events (using our notation)
\begin{equation}
\begin{array}{c}
\left(A_{i_{1}j_{1}}=a_{1},B_{i_{1}j_{1}}=b_{1}\right),\left(A_{i_{2}j_{2}}=a_{2},B_{i_{2}j_{2}}=b_{2}\right),\ldots,\left(A_{i_{n}j_{n}}=a_{n},B_{i_{n}j_{n}}=b_{n}\right),\\
\\
\left(A_{i'_{1}j'_{1}}=a'_{1},B_{i'_{1}j'_{1}}=b'_{1}\right),\left(A_{i'_{2}j'_{2}}=a'_{1},B_{i'_{2}j'_{2}}=b'_{2}\right),\ldots,\left(A_{i'_{n}j'_{n}}=a'_{n},B_{i'_{n}j_{'n}}=b'_{n}\right),\\
\ldots
\end{array}
\end{equation}
Two such sequences are called (mutually) exclusive if, for at least
one $k\in\left\{ 1,2,\ldots,n\right\} $, either $i_{k}=i'_{k}$ and
$a_{k}\not=a'_{k}$ or $j_{k}=j'_{k}$ and $b_{k}\not=b'_{k}$. Cabello
postulates then that
\begin{equation}
\sum\Pr\left[A_{i_{1}j_{1}}=a_{1},B_{i_{1}j_{1}}=b_{1}\right]\Pr\left[A_{i_{2}j_{2}}=a_{2},B_{i_{2}j_{2}}=b_{2}\right]\ldots\Pr\left[A_{i_{n}j_{n}}=a_{n},B_{i_{n}j_{n}}=b_{n}\right]\leq1
\end{equation}
if the sum is taken over any set of pairwise exclusive sequence, for
any $n$. This postulate allows him to successfully derive a certain
QM inequality {[}\ref{enu:A.A.-Klyachko,-M.A.}{]}, and he conjectures
that the Tsirelson bounds follow from this postulate too. This conjecture,
however, remains unproven, and the relation of Cabello's postulate
to KPT in general and to our characterization of the Tsirelson bounds
by means of compatible connections remains unclear. 

It may be useful to compare our approach to constructing a Kolmogovian
account of the EPR paradigm to the only other systematic way of doing
this known to us. We call it \emph{conditionalization}. It consists
in considering the settings $\left(\alpha_{i},\beta_{j}\right)$ as
values of a random variable $C$, and treating the spins as random
variables whose distributions are conditoned upon the values of $C$.
Avis, Fischer, Hilbert, and Khrennikov {[}\emph{\ref{enu:D.-Avis,-P.}}{]}
implement this approach by considering the system of jointly distributed
$\left(C,A'_{1},A'_{2},B'_{1},B'_{2}\right)$ such that
\begin{equation}
\Pr\left[A'_{i}=\pm1,B'_{j}=\pm1\,|\, C=\left(\alpha_{i},\beta_{j}\right)\right]=\Pr\left[A_{ij}=\pm1,B_{ij}=\pm1\right].
\end{equation}
They describe two ways of achieving this. In one of them $A'_{i}$
and $B'_{j}$ have three possible values, $\pm1$ and 0 (that can
interpreted as ``no value''), and
\begin{equation}
\Pr\left[A'_{i}=a,B'_{j}=b,\, A'_{3-i}=a',B'_{3-j}=b',|\, C=\left(\alpha_{i},\beta_{j}\right)\right]=\begin{cases}
\Pr\left[A_{ij}=a,B_{ij}=b\right] & \textnormal{if }a\not=0,b\not=0,a'=b'=0,\\
0 & \textnormal{otherwise}.
\end{cases}
\end{equation}
In another implementation, $A'_{i}$ and $B'_{j}$ have two possible
values, $\pm1$, and
\begin{equation}
\Pr\left[A'_{i}=a,B'_{j}=b,\, A'_{3-i}=a',B'_{3-j}=b',|\, C=\left(\alpha_{i},\beta_{j}\right)\right]=\frac{1}{4}\Pr\left[A_{ij}=a,B_{ij}=b\right].
\end{equation}
In both cases the joint distribution of $\left(C,A'_{1},A'_{2},B'_{1},B'_{2}\right)$
is well-defined for any distribution $C$ with non-zero values of
$\Pr\left[C=\left(\alpha_{i},\beta_{j}\right)\right]$, $i,j\in\left\{ 1,2\right\} $. 

An even simpler implementation of conditionalization would be to introduce
just two binary ($+1/-1$) random variables $A',B'$, and to construct
a joint distribution of $\left(C,A',B'\right)$ by positing
\begin{equation}
\Pr\left[A'=a,B'=b\,|\, C=\left(\alpha_{i},\beta_{j}\right)\right]=\Pr\left[A_{ij}=a,B_{ij}=b\right].
\end{equation}

Conditionalization is universally applicable, and it indeed achieves
the goal of embedding all imaginable distributions of $\left(A_{ij},B_{ij}\right)$,
$i,j\in\left\{ 1,2\right\} $, into the framework of KPT. In fact,
it veridically describes the experiment in which the settings $\left(\alpha_{i},\beta_{j}\right)$
are chosen randomly according to some distribution. As we argue in
greater detail elsewhere {[}\ref{enu:E.N.-Dzhafarov,-Conditionalization}{]},
however, this approach has its weakness: it is not only universal,
it is also indiscriminate. Conditionalization applies in precisely
the same way to the distributions of $\left(A_{ij},B_{ij}\right)$
whether they are subject to classical-mechanical constraints, to QM
constraints, or anything else, the choice of the distribution for
$C$ being irrelevant. The conditionalization approach therefore can
be compared to saying that the four pairs $\left(A_{ij},B_{ij}\right)$
can always be p-coupled as stochastically independent pairs: this
is true, but not elucidating. By contrast, our contextual approach
is aimed at characterizing different constraints imposed on the distributions
of $\left(A_{ij},B_{ij}\right)$ by their compatibility with different
distributions of the connections $\left(A_{i1},A_{i2}\right),\left(B_{1j},B_{2j}\right)$,
$i,j\in\left\{ 1,2\right\} $ (or other marginals). This allows us,
in particular, to characterize the classical-mechanical and Tsirelson
constraints, and to identify the QM constraint as falling beyond the
reach of such characterization.

\subsection*{Acknowledgments}

This work was supported by NSF grant SES-1155956.

\section*{References}
\begin{enumerate}
\item \label{enu:Bell,-J.-(1964).}J. Bell: On the Einstein-Podolsky-Rosen
paradox. \emph{Physics} \textbf{1}, 195-200 (1964).
\item \label{enu:J.F.-Clauser,-M.A.}J.F. Clauser, M.A. Horne, A. Shimony,
R.A. Holt: Proposed experiment to test local hidden-variable theories.
\emph{Phys. Rev. Lett.} \textbf{23}, 880-884 (1969).
\item \label{enu:Clauser,-J.F.-and}J.F. Clauser, M.A. Horne: Experimental
consequences of objective local theories\emph{. Phys. Rev. D}, \textbf{10},
526-535 (1974).
\item \label{enu:Fine,-A.-(1982b).}A. Fine: Hidden variables, joint probability,
and the Bell inequalities.\emph{ Phys. Rev. Lett.} \textbf{48}, 291-295
(1982).
\item \label{enu:A.-Kolmogorov,-Foundations}A. Kolmogorov: \emph{Foundations
of the Theory of Probability} (New York, Chelsea, 1956). 
\item \label{enu:S.-Kochen,-F.}S. Kochen, F. Specker: The problem of hidden
variables in quantum mechanics. J. Math. Mech. \textbf{17}, 59\textendash{}87
(1967).
\item \label{enu:F.-Laudisa-(1997)}F. Laudisa: Contextualism and nonlocality
in the algebra of EPR observables. \emph{Phil. Sci.} \textbf{64},
478-496 (1997).
\item \label{enu:R.W.-Spekkens,-D.}R.W. Spekkens, D. H. Buzacott, A. J.
Keehn, B. Toner, G. J. Pryde: Universality of state-independent violation
of correlation inequalities for noncontextual theories. \emph{Phys.
Rev. Lett.} \textbf{102}, 010401 (2009).
\item \label{enu:G.-Kirchmair,-F.}G. Kirchmair, F. Zähringer, R. Gerritsma,
M. Kleinmann, O. Gühne, A. Cabello, R. Blatt, C. F. Roos: State-independent
experimental test of quantum contextuality. \emph{Nature} \textbf{460},
494-497 (2009).
\item \label{enu:P.-Badzia=000327g1,-I.}P. Badzi\c{a}g1, I. Bengtsson1,
A. Cabello, I. Pitowsky: Universality of state-independent violation
of correlation inequalities for noncontextual theories. \emph{Phys.
Rev. Lett.} \textbf{103}, 050401 (2009).
\item \label{enu:A.Yu.-Khrennikov:-Contextual}A.Yu. Khrennikov: \emph{Contextual
Approach to Quantum Formalism. Fundamental Theories of Physics} \textbf{160}
(Dordrecht, Springer, 2009).
\item \label{enu:A.-Cabello:-Simple}A. Cabello: Simple explanation of the
quantum violation of a fundamental inequality. \emph{Phys. Rev. Lett.
}\textbf{110}, 060402 (2013).
\item \label{enu:E.N.-Dzhafarov,-J.V.PLOS}E.N. Dzhafarov, J.V. Kujala:
All-possible-couplings approach to measuring probabilistic context.
\emph{PLoS ONE 8(5): e61712. doi:10.1371/ journal.pone.0061712} (2013).
\item \label{enu:E.N.-Dzhafarov,-J.V.Qualified}E.N. Dzhafarov, J.V. Kujala:
A qualified Kolmogorovian account of probabilistic contextuality.
\emph{Lect. Notes in Comp. Sci}. (in press). (available as arXiv:1304.4546.)
\item \label{Dzhafarov,-E.N.,-&Advances}E.N. Dzhafarov, J.V. Kujala: Random
variables recorded under mutually exclusive conditions: Contextuality-by-Default.
\emph{Adv. in Cogn. Neurodyn. IV}. (in press). (available as arXiv:1309.0962.)
\item \label{enu:D.-Bohm,-&}D. Bohm, Y. Aharonov: Discussion of experimental
proof for the paradox of Einstein, Rosen and Podolski. \emph{Phys.
Rev}., \textbf{108}, 1070-1076 (1957).
\item \label{enu:E.N.-Dzhafarov,-J.V.}E.N. Dzhafarov, J.V. Kujala: On selective
influences, marginal selectivity, and Bell/CHSH inequalities. \emph{Topics
Cog. Sci.} (in press). (available as arXiv:1211.2342.)
\item \label{enu:A.-Yu.-Khrennikov,}A.Yu. Khrennikov: Bell-Boole inequality:
Nonlocality or probabilistic incompatibility of random variables?
\emph{Entropy} \textbf{10}, 19-32 (2008).
\item \label{enu:A.-Yu.-Khrennikov,-1}A.Yu. Khrennikov: EPR\textendash{}Bohm
experiment and Bell\textquoteright{}s inequality: Quantum physics
meets probability theory. \emph{Theor. Math. Phys.}\textbf{ 157},
1448\textendash{}1460 (2008).
\item \label{enu:S.-Gudder:-On}S. Gudder: On hidden-variable theories.
\emph{J. Math Phys.} \textbf{2}, 431-436 (1970).
\item \label{enu:H.-Thorisson,-Coupling,}H. Thorisson:\emph{ Coupling,
Stationarity, and Regeneration} (New York, Springer, 2000).
\item \label{enu:Landau-LJ-(1988)}L.J. Landau: Empirical two-point correlation
functions. \emph{Found. Phys.} \textbf{18}, 449--460 (1988).
\item \label{enu:Cabello-A-(2005)}A. Cabello: How much larger quantum correlations
are than classical ones. \emph{Phys. Rev.} A \textbf{72}, 12113 (2005).
\item \label{enu:Kujala,-J.-V.,2008}J.V. Kujala, E.N. Dzhafarov: Testing
for selectivity in the dependence of random variables on external
factors.\emph{ J. Math. Psych}. \textbf{52}, 128-144 (2008).
\item \label{enu:Cirel'son-BS-(1980)}B.S. Tsirelson: Quantum generalizations
of Bell's inequality.\emph{ Lett. Math. Phys.} \textbf{4}, 93--100
(1980).
\item \label{enu:Landau-LJ-(1987)}L.J. Landau: On the violation of Bell's
inequality in quantum theory. \emph{Phys. Lett. A} \textbf{120}, 54--56
(1987).
\item \label{enu:E.N.-Dzhafarov,-in press 2}E.N. Dzhafarov, J.V. Kujala:
Selectivity in probabilistic causality: Where psychology runs into
quantum physics. \emph{J. Math. Psych}., \textbf{56}, 54-63 (2012).
\item \label{enu:E.N.-Dzhafarov,-J.V._LNCS1}E.N. Dzhafarov, J.V. Kujala:
Quantum entanglement and the issue of selective influences in psychology:
An overview. \emph{Lect. Notes in Comp. Sci}. \textbf{7620}: 184-195
(2012).
\item \label{enu:E.N.-Dzhafarov,-&in press}E.N. Dzhafarov, J.V. Kujala:
Order-distance and other metric-like functions on jointly distributed
random variables. \emph{Proc. Amer. Math. Soc}. \textbf{141}, 3291-3301
(2013).
\item \label{enu:A.A.-Klyachko,-M.A.}A.A. Klyachko, M.A. Can, S. Binicioglu,
A.S. Shumovsky: Simple test for hidden variables in spin-1 systems.
\emph{Phys. Rev. Lett. }\textbf{101}, 020403 (2008).
\item \label{enu:D.-Avis,-P.}D. Avis, P. Fischer, A. Hilbert, A. Khrennikov:
Single, complete, probability spaces consistent with EPR-Bohm-Bell
experimental data. In A. Khrennikov (Ed) \emph{Foundations of Probability
and Physics-5, AIP Conference Proceedings} \textbf{750}: 294-301 (Melville,
New York, AIP, 2009).
\item \label{enu:E.N.-Dzhafarov,-Conditionalization}E.N. Dzhafarov, J.V.
Kujala: Embedding quantum into classical: Contextualization vs conditionalization.
arXiv:1312.0097 (2013).\end{enumerate}

\end{document}